\documentclass[a4paper,10pt]{article}
\usepackage{graphicx}
\usepackage{fullpage}
\usepackage[T1]{fontenc}
\usepackage{epsfig}
\usepackage{color}
\usepackage{url}

\usepackage{comment}
\usepackage{mathtools}
\usepackage{breqn}
\usepackage{amsmath,amssymb,amstext, amsthm}
\usepackage{amsbsy,bm}
\usepackage{amsfonts}
\usepackage[utf8]{inputenc}
\usepackage[shortlabels]{enumitem}
\setlist[itemize]{topsep=7pt, listparindent=\parindent}
\usepackage{verbatim}
\usepackage{multirow}
\usepackage{lineno}
\newcommand{\ceil}[1]{\left \lceil #1 \right \rceil}

\newcommand{\RMin}{\textsf{RMin}}
\newcommand{\RMax}{\textsf{RMax}}
\newcommand{\PSV}{\textsf{PSV}}
\newcommand{\PLV}{\textsf{PLV}}
\newcommand{\NSV}{\textsf{NSV}}
\newcommand{\NLV}{\textsf{NLV}}
\newcommand{\Min}{\textsf{Min}}
\newcommand{\Max}{\textsf{Max}}
\newcommand{\cMin}{\textsf{cMin}}
\newcommand{\cMax}{\textsf{cMax}}
\newcommand{\rank}{\textsf{rank}}
\newcommand{\select}{\textsf{select}}
\newcommand{\BP}{\textsf{BP}}
\newcommand{\colorr}{\textsf{color}}
\newcommand{\PRS}{\textsf{PRS}}
\newcommand{\NRS}{\textsf{NRS}}

\newtheorem{theorem}{Theorem}
\newtheorem{lemma}{Lemma}%
\newtheorem{corollary}{Corollary}

\newtheorem{remark}{Remark}%

\begin{document}

\title{Space-efficient Data Structure for Next/Previous Larger/Smaller Value Queries \footnote{Preliminary version of these results appeared in the proceedings of the The 15th Latin American Theoretical Informatics Symposium (LATIN 2022)~\cite{DBLP:conf/latin/JoK22}. This work was supported by the National Research Foundation of Korea (NRF) grant funded by the Korea government (MSIT) (No. NRF-2020R1G1A1101477). Seungbum Jo and Geunho Kim contributed equally to this work. We would like to thank to Srinivasa Rao Satti for helpful discussions.}}

\author{
    Seungbum Jo \\
    Chungnam National University, South Korea \\
    sbjo@cnu.ac.kr
    \and
    Geunho Kim \\
    The University of Tokyo, Japan  \\
    research@gnho.kim
}
\date{}

\maketitle
\abstract{Given an array of size $n$ from a total order, we consider the problem of constructing a data structure that supports various queries (range minimum/maximum queries with their variants and next/previous larger/smaller queries) efficiently. In the encoding model (i.e., the queries can be answered without the input array), we propose a $(3.701n + o(n))$-bit data structure, which supports all these queries in $O(\log^{(\ell)}n)$ time, for any positive constant integer $\ell$ (here, $\log^{(1)} n  = \log n$, and for $\ell > 1$, $\log^{(\ell)} n = \log ({\log^{(\ell-1)}} n)$). The space of our data structure matches the current best upper bound of Tsur (Inf. Process. Lett., 2019), which does not support the queries efficiently.
Also, we show that at least $3.16n-\Theta(\log n)$ bits are necessary for answering all the queries. Our result is obtained by generalizing Gawrychowski and Nicholson's $(3n - \Theta(\log n))$-bit lower bound  (ICALP, 15) for answering range minimum and maximum queries on a permutation of size $n$.}



\section{Introduction}
Given an array $A[1, \dots, n]$ of size $n$ from a total order and an interval $[i, j] \subset [1,n]$, suppose there are $k$ distinct positions $i \le p_1 \le p_2 \dots \le p_k \le j$ where $p_1, p_2, \dots, p_{k}$ are the positions of minimum elements in $A[i, \dots, j]$. Then, for $q \ge 1$, \textit{range $q$-th minimum query} on the interval $[i, j]$ ($\RMin{}(i, j, q)$) returns the position $p_q$ (returns $p_k$ if $q > k$), and \textit{range minimum query} on the interval $[i, j]$ ($\RMin{}(i, j)$) returns an arbitrary position among $p_1, p_2, \dots, p_{k}$.
One can also analogously define \textit{range $q$-th maximum query} (resp. \textit{range maximum query}) on the interval $[i, j]$, denoted by $\RMax{}(i, j, q)$ (resp. $\RMax{}(i, j)$).

In addition to the above queries, one can define next/previous larger/smaller queries as follows. When the position $i$ is given, the \textit{previous smaller value query}  on the position $i$ ($\PSV{}(i)$) returns the rightmost position $j < i$, where $A[j]$ is smaller than $A[i]$ (returns $0$ if no such $j$ exists), and the \textit{next smaller value query} on the position $i$ ($\NSV{}(i)$)  returns the leftmost position $j > i$ where $A[j]$ is smaller than $A[i]$ (returns $n+1$ if no such $j$ exists).
The \textit{previous (resp. next) larger value query} on the position $i$, denoted by $\PLV{}(i)$ (resp. $\NLV{}(i)$)) is also defined analogously.

In this paper, we focus on the problem of constructing a data structure that efficiently answers all the above queries. We consider the problem in the \textit{encoding model}~\cite{DBLP:conf/walcom/Raman15}, which does not allow access to the input $A$ for answering the queries after preprossessing.
In the encoding data structure, the lower bound of the space is referred to as the \textit{effective entropy} of the problem. Note that for many problems, their effective entropies have much smaller size compared to the size of the inputs~\cite{DBLP:conf/walcom/Raman15}. Also, an encoding data structure is called \textit{succinct} if its space usage matches the optimal up to lower-order additive terms. The rest of the paper only considers encoding data structures and assumes a $\Theta (\log n)$-bit word RAM model, where $n$ is the input size.

\subsection{Previous Work} 
The problem of constructing an encoding data structure for answering range minimum queries has been well-studied because of its wide applications. It is well-known that any two arrays have a different set of answers of range minimum queries if and only if their corresponding Cartesian trees~\cite{DBLP:journals/cacm/Vuillemin80} are distinct. Thus, the effective entropy of answering range minimum queries on the array $A$ of size $n$ is $2n-\Theta(\log n)$ bits. Sadakane~\cite{DBLP:journals/jda/Sadakane07} proposed the $(4n+o(n))$-bit encoding with $O(1)$ query time using the balanced-parenthesis (BP)~\cite{DBLP:journals/siamcomp/MunroR01} of the Cartesian tree on $A$ with additional nodes. Fisher and Heun~\cite{DBLP:journals/siamcomp/FischerH11} proposed the $(2n+o(n))$-bit data structure (hence, succinct), which supports $O(1)$ query time using the depth-first unary degree sequence (DFUDS)~\cite{DBLP:journals/algorithmica/BenoitDMRRR05} of the \textit{2d-min heap} on $A$. Here, a 2d-min heap of $A$ is an alternative representation of the Cartesian tree on $A$.
By maintaining the encodings of both 2d-min and max heaps on $A$ (2d-max heap can be defined analogously to 2d-min heap), the encoding of \cite{DBLP:journals/siamcomp/FischerH11} directly gives a $(4n+o(n))$-bit encoding for answering both range minimum and maximum queries in $O(1)$ time. Gawrychowski and Nicholson~\cite{DBLP:conf/icalp/GawrychowskiN15} reduced this space to $(3n+o(n))$-bit while supporting the same query time for both queries. They also showed that the effective entropy for answering the range minimum and maximum queries is at least $3n-\Theta(\log n)$ bits.

Next/previous smaller value queries were motivated from the parallel computing~\cite{DBLP:journals/jal/BerkmanSV93}, and have application in constructing compressed suffix trees~\cite{DBLP:conf/spire/OhlebuschFG10}. If all elements in $A$ are distinct, one can answer both the next and previous smaller queries using Fischer and Heun's encoding for answering range minimum queries~\cite{DBLP:journals/siamcomp/FischerH11}. For the general case, Ohlebusch et al.~\cite{DBLP:conf/spire/OhlebuschFG10} proposed the $(3n+o(n))$-bit encoding for supporting range minimum and next/previous smaller value queries in $O(1)$ time. Fischer~\cite{DBLP:journals/tcs/Fischer11} improved the space to $2.54n+o(n)$ bits while maintaining the same query time. 
More precisely, their data structure uses the \textit{colored 2d-min heap} on $A$, which is a 2d-min heap on $A$ with a bichromatic coloring on its nodes. Fischer~\cite{DBLP:journals/tcs/Fischer11} showed that the proposed data structure is succinct by proving that two arrays with distinct colored 2d-min heaps yield the different query answers (the effective entropy of the colored 2d-min heap on $A$ is $2.54n-\Theta(\log n)$ bits~\cite{DBLP:journals/dam/MerliniSV04}).
For any $q \ge 1$, the encoding of \cite{DBLP:journals/tcs/Fischer11} also supports the range $q$-th minimum queries in $O(1)$ time~\cite{DBLP:journals/tcs/JoS16}.

From the above, the encoding of Fischer~\cite{DBLP:journals/tcs/Fischer11} directly gives a $(5.08n+o(n))$-bit data structure for answering the range $q$-th minimum/maximum queries and next/previous larger/smaller value queries in $O(1)$ time by maintaining the data structures of both colored 2d-min and max heaps. Jo and Satti~\cite{DBLP:journals/tcs/JoS16} improved the space to (i) $4n+o(n)$ bits if there are no consecutive equal elements in $A$ and (ii) $4.585n+o(n)$ bits for the general case while supporting all the queries in $O(1)$ time. They also showed that if the query time is not of concern, the space of (ii) can be improved to $4.088n+o(n)$ bits. Recently, Tsur~\cite{DBLP:journals/ipl/Tsur19} improved the space to $3.585n$ bits if there are no consecutive equal elements in $A$ and $3.701n$ bits for the general case. However, their encoding does not support the queries efficiently ($O(n)$ time for all queries).

\subsection{Our Results}
Given an array $A[1, \dots, n]$ of size $n$ with the interval $[i, j] \subset [1, n]$ and the position $1 \le p \le n$, we show the following results for any $q \ge 1$:
\begin{enumerate}[(a)]
\item If $A$ has no two consecutive equal elements, there exists a $(3.585n+o(n))$-bit data structure, which can answer (i) $\RMin{}(i, j)$, $\RMax{}(i, j)$, $\PSV{}(p)$, and $\PLV{}(p)$ queries in $O(1)$ time, and (ii) $\RMin{}(i, j, q)$, $\RMax{}(i, j, q)$, $\NSV{}(p)$, and $\NLV{}(p)$ queries in $O(\log^{(\ell)} n)$ time\footnote{Throughout the paper, we denote $\log n$ as the logarithm to the base $2$}, for any positive constant integer $\ell$.
\item For the general case, the data structure of (a) uses $3.701n+o(n)$ bits while supporting the same query time.
\end{enumerate}
Our results match the current best upper bounds of Tsur~\cite{DBLP:journals/ipl/Tsur19} up to lower-order additive terms while supporting the queries efficiently. The main idea of our encoding data structure is to combine the BP of colored 2d-min and max heap of $A$. Note that all previous encodings in \cite{DBLP:conf/icalp/GawrychowskiN15, DBLP:journals/tcs/JoS16, DBLP:journals/ipl/Tsur19} combine the DFUDS of the (colored) 2d-min and max heap on $A$.

\begin{table}[t]
	\caption{Summary of the upper and lower bounds results of encoding data structures for answering $q$-th minimum/maximum queries and previous/next larger/smaller value queries on the array $A[1, \dots, n]$, for any $q \ge 1$ ($\ell$ is as any positive constant integer). Note that all our upper bound results support range minimum/maximum and previous larger/smaller value queries in $O(1)$ time.} 
\centering
\scalebox{0.8}{
		\begin{tabular}{cccc}
			\hline
			Array Type & Space (in bits) & Query time & Reference \\
			\hline
			\multicolumn{4}{c}{Upper bounds}\\
			\hline
			\multirow{3}{*}{$A[i] \neq A[i+1]$ for all $i \in [1, n-1]$}&  $4n+o(n)$ &  $O(1)$  & ~\cite{DBLP:journals/tcs/JoS16}\\
			&  $3.585n$ & $O(n)$  & ~\cite{DBLP:journals/ipl/Tsur19}\\
			& $3.585n + o(n)$ &  $O(\log^{(\ell)} n)$  & Theorem~\ref{thm:nonconsecutive}\\
			\hline
			\multirow{6}{*}{General array} & $5.08n+o(n)$ & $O(1)$ &~\cite{DBLP:journals/tcs/Fischer11, DBLP:journals/tcs/JoS16}\\ 
			&  $4.088n+o(n)$ & $O(n)$  & 	\multirow{2}{*}{~\cite{DBLP:journals/tcs/JoS16}} \\
			&  $4.585n+o(n)$ & $O(1)$  & \\
			&  $3.701n$ & $O(n)$  & ~\cite{DBLP:journals/ipl/Tsur19}\\
			&  $3.701n+o(n)$ & $O(\log^{(\ell)} n)$  & Theorem~\ref{thm:general}\\
			\hline
			\multicolumn{4}{c}{Lower bounds}\\
			\hline
			Permutation &$3n-\Theta(\log n)$&&~\cite{DBLP:conf/icalp/GawrychowskiN15}\\
			General array& $3.16n-\Theta(\log n)$ && Theorem~\ref{thm:lb}\\
			\hline
	\end{tabular}}
	\label{tab:summary}
\end{table}

\begin{table}[t]
	\caption{Summary of the upper and lower bounds results of encoding data structures for answering $q$-th minimum queries and previous/next smaller value queries on the array $A[1, \dots, n]$, for any $q \ge 1$ ($\ell$ is as any positive constant integer). Here $d_1$ denotes the number of positions $i \in \{2, \dots, n\}$ in $A$ which satisfy $\PSV{}(i-1) = \PSV{}(i)$, after removing all consecutive equal elements in $A$. Note that all our upper bound results support range minimum and previous smaller value queries in $O(1)$ time.} 
\centering
\scalebox{0.8}{
		\begin{tabular}{cccc}
			\hline
			Array Type & Space (in bits) & Query time & Reference \\
			\hline
			\multicolumn{4}{c}{Upper bounds}\\
			\hline
			{$A[i] \neq A[i+1]$ for all $i \in [1, n-1]$}&  $2.585n-d_1+o(n)$ & $O(\log^{(\ell)} n)$  & Corollary~\ref{cor:nonconsecutive}\\
			\hline
			\multirow{5}{*}{General array} & $2.54n+o(n)$ & $O(1)$ &~\cite{DBLP:journals/tcs/Fischer11}\\
			&  $3n+o(n)$ & $O(1)$  & 	~\cite{DBLP:journals/tcs/JoS16} \\
			&  $2.808n-d_1+o(n)$ & $O(\log^{(\ell)} n)$  & Corollary~\ref{cor:general}\\
             &  $2.585n+o(n)$ & $O(\log^{(\ell)} n)$  & Theorem~\ref{thm:cmin_second}\\
			\hline
			\multicolumn{4}{c}{Lower bounds}\\
			\hline
			General array &$2.54n-\Theta(\log n)$&&~\cite{DBLP:journals/tcs/Fischer11}\\
			\hline
	\end{tabular}}
	\label{tab:summary2}
\end{table}

We first consider the case when the array $A$ has no two consecutive equal elements (Section~\ref{sec:noequal}). In this case, we show that by storing the BP of the colored 2d-min heap on $A$ along with the color information of the nodes, there exists an encoding that supports range minimum, range $q$-th minimum, and next/previous smaller value queries efficiently. 
Compared to the data structure of~\cite{DBLP:journals/tcs/JoS16} that uses DFUDS of colored 2d-min heap, the encoding can use less space for storing the color information when $A$ has no consecutive equal elements. 

Next in section~\ref{sec:combine:nonc}, we describe how to combine the data structures on colored 2d-min and max heap on $A$ into a single structure of size at most $3.585n+o(n)$ bits to support range $q$-th minimum/maximum queries and previous/next larger/smaller value queries efficiently. The combined data structure is motivated by the idea of Gawrychowski and Nicholson’s encoding~\cite{DBLP:conf/icalp/GawrychowskiN15} to combine the DFUDS of 2d-min and max heap on $A$. As a consequence, we show that the combined data structure can be easily modified to a data structure of size at most $2.585n+o(n)$ bits, which supports range minimum, range $q$-th minimum, and next/previous smaller value queries in the same time as the combined data structure.

In Section~\ref{sec:equal}, we consider the case that $A$ contains consecutive equal elements. In this case, we show that there exists a data structure of at most $3.701n + o(n)$ bits that supports range $q$-th minimum/maximum queries and previous/next larger/smaller value queries on $A$ efficiently. 
The main idea of the data structure is to combine the data structure on the array $A'$, which is obtained by removing all consecutive equal elements from $A$, with some additional auxiliary structures. 
Again, the data structure can be modified to the data structure of size at most $2.808n + o(n)$ bits, which supports range minimum, range $q$-th minimum, and next/previous smaller value queries in the same time. Furthermore, we show that the worst-case space usage of the data structure can be improved to $2.585n+o(n)$ bits with the same query time. Compare to the Fischer's $(2.54n+o(n))$-bit succinct encoding~\cite{DBLP:journals/tcs/Fischer11}, our data structure takes slightly more space for the worst case input, with slower query time for $q$-th minimum and next smaller value queries. 
However, the space usage of our data structures is dependent on two factors: (i) the number of consecutive equal elements, and (ii) the length of the increasing or decreasing runs in the input. In contrast, the data structure proposed by Fischer~\cite{DBLP:journals/tcs/Fischer11} maintains a fixed space usage regardless of the input. As a result, our data structures may outperform Fischer's data structure and even take less space than the information-theoretical lower bound in certain inputs. 

Finally, in Section~\ref{sec:lower}, we show that the effective entropy of the encoding to support the range $q$-th minimum and maximum queries on $A$ is at least $3.16n-\Theta(\log n)$ bits. Our result is obtained by extending the $(3n-\Theta(\log n))$-bit lower bound of Gawrychowski and Nicholson~\cite{DBLP:conf/icalp/GawrychowskiN15} for answering the range minimum and maximum queries on a permutation of size $n$.
We summarize our results in Table~\ref{tab:summary}.

\section{Preliminaries}
This section introduces some data structures used in our results.
\begin{figure}
\begin{center}
\includegraphics[scale=0.28]{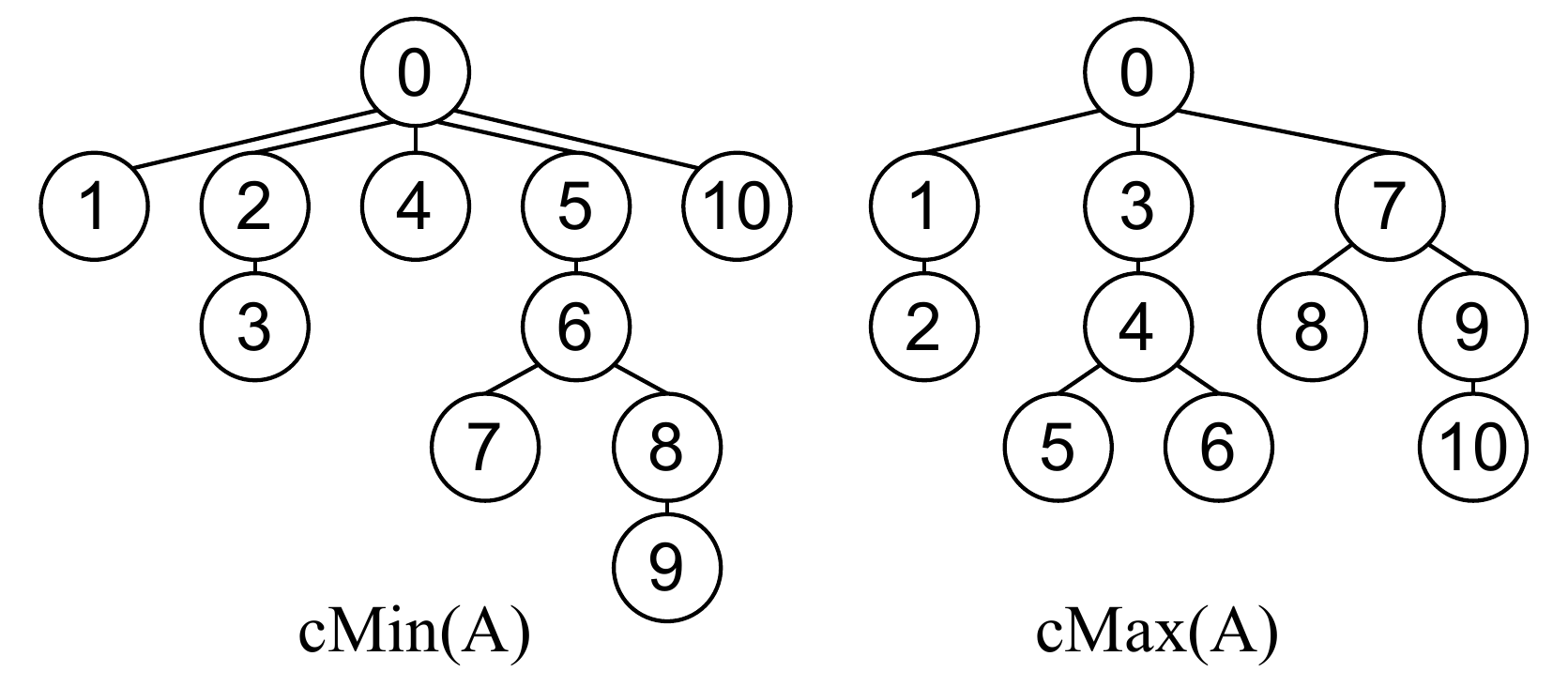}
\end{center}
\caption{$\Min{}(A)$ and $\Max{}(A)$ on the array $A =~5~4~5~3~1~2~6~3~4~1$.}
\label{fig:minheap}
\end{figure}
\\\\
\noindent\textbf{2d min-heap and max-heap. } Given an array $A[1, \dots, n]$ of size $n$, the 2d min-heap on $A$ (denoted by $\Min(A)$)~\cite{DBLP:journals/tcs/Fischer11} is a rooted and ordered tree with $n+1$ nodes, where each node corresponds to the value in $A$, and the children are ordered from left to right. More precisely, $\Min(A)$ is defined as follows:
\begin{enumerate}
    \item The root of $\Min(A)$ corresponds to $A[0]$ ($A[0]$ is defined as $-\infty$). 
    \item For any $i > 0$, $A[i]$ corresponds to the $(i+1)$-th node of $\Min(A)$ according to the preorder traversal. 
    \item For any non-root node corresponds to $A[j]$, its parent node corresponds to $A[\PSV{}(j)]$.
\end{enumerate}

In the rest of the paper, we refer to the node $i$ in $\Min{}(A)$ as the node corresponding to $A[i]$ (i.e., the $(i+1)$-th node according to the preorder traversal).
One can also define the 2d-max heap on $A$ (denoted as $\Max{}(A)$) analogously. More specifically, in $\Max{}(A)$, $A[0]$ is defined as $\infty$, and the parent of node $i > 0$ corresponds to the node $\PLV{}(i)$ (see Figure~\ref{fig:minheap} for an example). In the rest of the paper, we only consider $\Min{}(A)$ unless $\Max{}(A)$ is explicitly mentioned. The same definitions and properties for $\Min{}(A)$ can be applied to $\Max{}(A)$. 
From the definition of $\Min{}(A)$, Tsur~\cite{DBLP:journals/ipl/Tsur19} showed the following lemma. 

\begin{lemma}[\cite{DBLP:journals/ipl/Tsur19}]\label{lem:2dmin}
For any $i \in \{1, 2, \dots, n-1\}$, the following holds:
\begin{enumerate}[(a)]
\item If the node $i$ is an internal node in $\Min{}(A)$, then the node $(i+1)$ is the leftmost child of the node $i$ in $\Min{}(A)$. 
\item If $A$ has no two consecutive equal elements, the node $i$ is an internal node in $\Min{}(A)$ (resp. $\Max{}(A)$) if and only if the node $i$ is a leaf node in $\Max{}(A)$ (resp. $\Min{}(A)$). \end{enumerate}
\end{lemma}

\begin{figure}
\begin{center}
\includegraphics[scale=0.28]{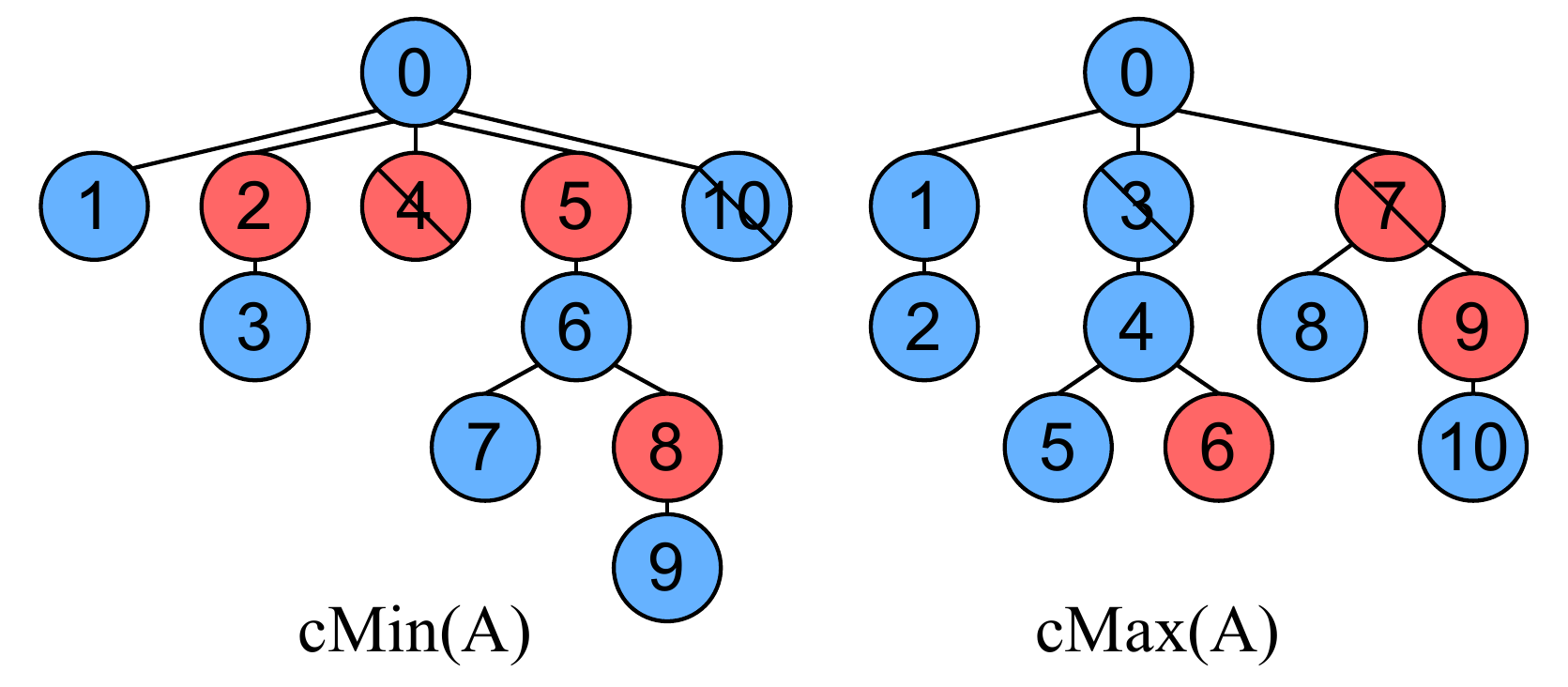}
\end{center}
\caption{$\cMin{}(A)$ and $\cMax{}(A)$ on the array $A =~5~4~5~3~1~2~6~3~4~1$. The nodes with slash line indicate the valid nodes.}
\label{fig:cminheap}
\end{figure}

\noindent\textbf{Colored 2d min-heap and max-heap. } The colored 2d-min heap of $A$ (denoted by $\cMin{}(A)$)~\cite{DBLP:journals/tcs/Fischer11} is $\Min{}(A)$ where each node is colored \textit{red} or \textit{blue} as follows. The node $i$ in $\cMin{}(A)$ is colored red if and only if $i$ is not the leftmost child of its parent node, and $A[i] \neq A[j]$, where the node $j$ is $i$'s immediate left sibling. Otherwise, the node $i$ is colored blue. One can also define the colored 2d-max heap on $A$ (denoted by $\cMax{}(A)$) analogously (see Figure~\ref{fig:cminheap} for an example). The following lemma says that we can obtain the color of some nodes in $\cMin{}(A)$ from their tree structures.

\begin{lemma}\label{lem:c2dmin}
For any node $i$ in $\cMin{}(A)$, the following holds:
\begin{enumerate}[(a)]
\item If the node $i$ is the leftmost child of its parent node, the color of the node $i$ is always blue.
\item If $A$ has no two consecutive equal elements, the color of the node $i$ is always red if its immediate left sibling is a leaf node.
\end{enumerate}
\begin{proof}
(a) is directly proved from the definition of $\cMin{}(A)$. Also, if the immediate left sibling $j$ of the node $i$ is a leaf node, $j$ is equal to $i-1$ (note that the preorder traversal of $\cMin{}(A)$ visits the node $i$ immediately after visiting the node $j$). Thus, if $A$ has no consecutive equal elements, the color of the node $i$ is red.
\end{proof}
\end{lemma}

We say a non-root node is called \text{valid} in $\cMin(A)$ if and only if it is neither the a leftmost child nor an immediate right sibling of a leaf node. Otherwise, the node $i$ is \textit{invalid}. By Lemma~\ref{lem:c2dmin}, if $A$ has no two consecutive equal elements, the color of the invalid nodes of $\cMin{}(A)$ can be decoded from the tree structure.
\\\\
\noindent\textbf{Rank and Select queries on bit arrays. } Given a bit array $B[1, \dots, n]$ of size $n$, and a pattern $p \in \{0, 1\}^{+}$, (i) $\rank{}_{p}(i, B)$ returns the number of occurrence of the pattern $p$ in $B[1, \dots, i]$, and (ii) $\select{}_{p}(j, B)$ returns the first position of the $j$-th occurrence of the pattern $p$ in $B$. The following lemma shows that there exists a succinct encoding, which supports both $\rank$ and $\select$ queries on $B$ efficiently.

\begin{table}[htp]
    \centering
    \scalebox{0.8}{
    \begin{tabular}{l p{14.5cm}}
    \hline
    \textbf{Symbol} & \textbf{Description} \\
    \hline
    \hline
    Section~\ref{sec:cmin} & \\
    \hline
    $c_{min}$ (resp. $c_{max}$) & Bit array that stores the color of all valid nodes in $\cMin{}(A)$ (resp. $\cMax{}(A)$). \\
    $M_{\ell'}$ & Bit array that marks the nodes in $\cMin(A)$ at the $\ell'$-th level. \\
    $cr(p)$ & Child rank of node $p$, i.e., the number of left siblings of $p$. \\
    $pre_{\ell'}(p)$  & The rightmost sibling of $p$ to the left that is marked at the $\ell'$-th level. \\
    $next_{\ell'}(p)$ & The leftmost sibling of $p$ to the right that is marked at the $\ell'$-th level. \\
    $P_{\ell'}$ & Bit array that, for each marked node $p$ at the $\ell'$-th level, stores the difference in child rank between $p$ and (i) the rightmost sibling to its left and (ii) the leftmost sibling to its right that, where each sibling is either the marked node at the $(\ell'-1)$-th level or the red-colored node.\\
     \hline
    Section~\ref{sec:combine:nonc} & \\
    \hline
    $U$ & Bit array that indicates whether a node is an internal node in $\cMin(A)$ or $\cMax(A)$.\\
    $S$ & Bit array that, for each leaf node $i$ in $T \in {\cMin(A), \cMax(A)}$, stores the number of $1$'s between positions $f(i, T)$ and $f(i+1, T)$.\\
    $D$ and $E$ & Ternary and bit arrays, respectively, that indicate the number of $1$'s between each pair of consecutive $0$'s in $S$. \\
    $d_1$ and $d_2$ & The number of $1$'s and $2$'s in $D$, respectively. \\
    $c_{minmax}$ & Bit array that concatenates $c_{min}$ and $c_{max}$. \\
    $B$ & Subarray $\BP(\cMin{}(A)$ of size $f(n, \BP(\cMin{}(A)) -1$. \\
    $f'(n)$ & $f(n, \BP(\cMin{}(A))-1$. \\
    $\alpha(j)$, $\beta(j)$, and $\gamma(j)$ & The positions in $U$, $D$, and $E$ corresponding to position $j$ in $B$, respectively. \\
    $U_1$, $D_1$, $E_1$ & Bit arrays that indicate the starting positions of each block in $U$, $D$, and $E$, respectively. \\
    $U_2$, $D_2$, $E_2$ & Bit arrays that indicate whether two consecutive blocks in $U$, $D$, and $E$ have the same starting positions, respectively. \\
     $F_i$ & Subsequence of $i$-th bad block of $E$. \\
    $g(u, d, e, b)$ & Function that reconstructs a subarray of $B$ from the subarrays $u$, $d$, and $e$ of $U$, $D$, and $E$, respectively, where $b \in \{0, 1\}$. \\
    $R$ & Array that stores the number of consecutive $1$'s from the beginning of each block in $B$. \\
      \hline
    Section~\ref{sec:equal} & \\
     \hline
    $C$ & Bit array that indicates positions of consecutive equal elements in $A$.\\
    $A'$ & Array that discards all consecutive equal elements in $A$. \\
    $B_A$ (resp. $B_{A'}$) & $\BP(\cMin{}(A))$ (resp. $\BP(\cMin{}(A'))$). \\
     $b'(j)$ & Position in $B_{A'}$ corresponding to the original position $j$ in $B_A$. \\
      $s_i$ & Starting position of the $i$-th block in $B_A$. \\
      $M_{B}$ (resp. $M_C$) & Bit arrays that indicate the corresponding positions of each $s_i$ in $B_{A'}$ (resp. $C$). \\
      $M'_{B}$ (resp. $M'_C$) & Bit arrays that indicate whether two consecutive positions $s_i$ and $s_{i+1}$ have the same corresponding positions in $B_{A'}$ (resp. $C$). \\
      $h(b, c)$ &  Function that reconstructs a subarray of $B_A$ from the subarrays $b$, $c$ of $B_{A'}$ and $C$, respectively. \\
      $M_S$ & Bit array that indicates the first bit of the $i$-th block of $B_{A}$ is $0$ or not. \\
      $m$ & The number of leftmost children in $\cMin{}(A)$.\\
      $I$ & Bit array that indicates the leftmost children in $\cMin{}(A)$.\\
      $I_c$ & Bit array that stores the colors of the nodes in $\cMin(A)$ that are not leftmost children.\\
      $B'$ & Bit array constructed by removing certain bits from $B_A$.\\
       $N_{B}$ (resp. $N_I$) & Bit arrays that indicate the corresponding positions of each $s_i$ in $B'$ (resp. $I$). \\
      $N'_{B}$ (resp. $N'_I$) & Bit arrays that indicate whether two consecutive positions $s_i$ and $s_{i+1}$ have the same corresponding positions in $B'$ (resp. $I$). \\
      $B_S$ & Bit array that stores the bits of $B_A$ at all positions $s_i$. \\
      $g_2(b', I_b)$ &  Function that reconstructs a subarray of $B_A$ from the subarrays $b'$, $I_b$ of $B'$ and $I$, respectively. \\
    \hline
    Section~\ref{sec:lower}& \\
     \hline
     $\mathcal{A}_n$ & Set of arrays of size $n$ constructed using the procedure described in Section~\ref{sec:lower}. \\
     $X$ & Set of $n - k$ positions that are not chosen during the construction of an array in $\mathcal{A}_n$.\\
    $\pi_{n-k}$ & Baxter permutation of size $n-k$.\\
    \hline
    \hline
    \end{tabular}}
    \caption{Summary of the main notations used in this article}
    \label{tab:symbols}
\end{table}

\begin{lemma}[\cite{DBLP:journals/jal/MunroRR01, DBLP:journals/talg/RamanRS07}]\label{lem:rrr}
Given a bit array $B[1, \dots, n]$ of size $n$ containing $m$ $1$s, and a pattern $p \in \{0, 1\}^{+}$ with $|p| \le \frac{\log n}{2}$,
$B$ can be stored in $(\log {n \choose m} + o(n))$ bits, supporting access to any $\Theta(\log n)$-sized consecutive bits of $B$ in $O(1)$ time. 
Furthermore, if one can access any $\Theta(\log n)$-sized consecutive bits of $B$ in $O(1)$ time, both $\rank{}_{p}(i, B)$ and $\select{}_{p}(j, B)$ queries can be answered in $O(1)$ time using $o(n)$-bit auxiliary structures.
\end{lemma}

\noindent\textbf{Balanced-parenthesis of trees. } 
Given a rooted and ordered tree $T$ with $n$ nodes, the balanced-parenthesis (BP) of $T$ (denoted by $\BP(T)$)~\cite{DBLP:journals/siamcomp/MunroR01} is a bit array defined as follows. We perform  a preorder traversal of $T$. We then add a $0$ to $\BP(T)$ when we first visit a node and add a $1$ to $\BP(T)$ after visiting all nodes in the subtree of the node. Since we add single $0$ and $1$ to $\BP(T)$ per each node in  $T$, the size of $\BP{}(T)$ is $2n$. For any node $i$ in $T$, we define $f(i, T)$ and $s(i, T)$ as the positions of the $0$ and $1$ in $\BP(T)$ that are added when node $i$ is visited during the traversal.
When $T$ is clear from the context, we write $f(i)$ (resp. $s(i)$) to denote $f(i, T)$ (resp. $s(i, T)$). If $T$ is a 2d-min heap, $f(i, T) = \select{}_{0}(i+1, BP(T))$ by the definition of 2d-min heap.

\noindent\textbf{Notations. } 
The remaining main notations used throughout the paper are summarized in Table~\ref{tab:symbols}.

\section{Data structure on arrays with no consecutive equal elements}\label{sec:noequal}
In this section, for any positive constant integer $\ell$, we present a $(3.585n + o(n))$-bit data structure on $A[1, \dots, n]$, which supports (i) range minimum/maximum and previous larger/smaller queries on $A$ in $O(1)$ time, and (ii) range $q$-th minimum/maximum and next larger/smaller value queries on $A$ in $O(\log^{(\ell)}n)$ time for any $q \ge 1$, when there are no two consecutive equal elements in $A$. We first describe the data structure on $\cMin{}(A)$ for answering the range minimum, range $q$-th minimum, and next/previous smaller value queries on $A$. Next, we show how to combine the data structures on $\cMin{}(A)$ and $\cMax{}(A)$ in a single structure. 

\subsection{Encoding data structure on $\cMin{}(A)$}\label{sec:cmin}
We store $\cMin{}(A)$ by storing its tree structure along with the color information of the nodes.
To store the tree structure, we use $\BP(\cMin{}(A))$. Also, for storing the color information of the nodes, we use a bit array $c_{min}$, which stores the color of all valid nodes in $\cMin{}(A)$ according to the preorder traversal order. In $c_{min}$ we use $0$ (resp. $1$) to indicate the color blue (resp. red). It is clear that $\cMin{}(A)$ can be reconstructed from $\BP{}(\cMin{}(A))$ and $c_{min}$.
Since $\BP{}(\cMin{}(A))$ and $c_{min}$ takes  $2(n+1)$ bits and at most $n$ bits, respectively, the total space for storing $\cMin{}(A)$ takes at most $3n + 2$ bits.
Note that a similar idea is used in Jo and Satti's \textit{extended DFUDS}~\cite{DBLP:journals/tcs/JoS16}, which uses the DFUDS of $\cMin{}(A)$ for storing the tree structure. However, extended DFUDS stores the color of all nodes other than the leftmost children, whereas $c_{min}$ does not store the color of all invalid nodes. The following lemma shows that from $\BP{}(\cMin{}(A))$, we can check whether the node $i$ is valid or not without decoding the entire tree structure.

\begin{lemma}\label{lem:valid}
When $A$ has no consecutive equal elements, the node $i$ is valid in $\cMin{}(A)$ if and only if $f(i) > 2$ and both $BP(\cMin{}(A))[f(i)-2]$ and $BP(\cMin{}(A))[f(i)-1]$ are $1$.
\begin{proof}
If both $BP(\cMin{}(A))[f(i)-2]$ and $BP(\cMin{}(A))[f(i)-1]$ are $1$, the preorder traversal of $\cMin{}(A)$ must complete the traversal of two subtrees consecutively just before visiting the node $i$ for the first time, which implies the node $i$'s immediate left sibling is not a leaf node (hence the node $i$ is valid).

Conversely, if $BP(\cMin{}(A))[f(i)-1] = 0$, the node $i-1$ is an internal node in $\cMin{}(A)$. Thus, the node $i$ is the leftmost child of the node $(i-1)$ by Lemma~\ref{lem:2dmin}. Next, if $BP(\cMin{}(A))[f(i)-2] = 0$ and $BP(\cMin{}(A))[f(i)-1] = 1$, the node $(i-1)$ is the immediate left sibling of the node $i$ since $f(i)-2$ is equal to $f(i-1)$. Also the node $(i-1)$ is a leaf node since $f(i)-1$ is equal to $s(i-1)$. Thus, the node $i$ is invalid in this case.
\end{proof}
\end{lemma}

Now we describe how to support range minimum, range $q$-th minimum, and next/previous smaller value queries efficiently on $A$ using $\BP{}(\cMin{}(A))$ and $c_{min}$ with $o(n)$-bit additional auxiliary structures. Note that both the range minimum and previous smaller value query on $A$ can be answered in $O(1)$ time using $BP(\cMin{}(A))$ with $o(n)$-bit auxiliary structures~\cite{DBLP:journals/talg/NavarroS14, DBLP:conf/dcc/FerradaN16}. Thus, it is enough to consider how to support a range $q$-th minimum and next smaller value queries on $A$. We introduce the following lemma of Jo and Satti~\cite{DBLP:journals/tcs/JoS16}, which shows that one can answer both queries with some navigational and color queries on $\cMin{}(A)$.

\begin{lemma}[Lemma 3.1 in \cite{DBLP:journals/tcs/JoS16}]\label{lem:qthoperation}
Given $\cMin{}(A)$, suppose there exists a data structure, which can answer (i) the tree navigational queries (next/previous sibling, subtree size, degree, level ancestor, child rank, child select, and parent\footnote{refer to Table 1 in \cite{DBLP:journals/talg/NavarroS14} for detailed definitions of the queries}), and (ii) the following color queries:
\begin{itemize}
    \item $\colorr{}(i)$: return the color of the node $i$.
    \item $\PRS{}(i)$: return the rightmost red sibling to the left of the node $i$.
    \item $\NRS{}(i)$: return the leftmost red sibling to the right of the node $i$.
\end{itemize}
Then for any $q \ge 1$, range $q$-th minimum, and the next smaller value queries on $A$ can be answered in $O(1)$ tree navigational queries along with $O(1)$ $\colorr{}(i)$, $\PRS{}(i)$, and $\NRS{}(i)$ queries on $\cMin{}(A)$.
\end{lemma}
\begin{proof}
In the proof of Lemma 3.1 in \cite{DBLP:journals/tcs/JoS16}, Jo and Satti described how to support range $q$-th minimum and next smaller value queries on $A$ in the paragraphs on $\textsf{RkMinQ}_A(i, j)$ and $\textsf{NSV}_A(i)$, respectively. They also introduced the array $V_{min}$ and the operation $node\_color$ to support $\colorr{}(i)$; other operations are analogous to those in this paper.
\end{proof}

Since all tree navigational queries in Lemma~\ref{lem:qthoperation} can be answered in $O(1)$ time using $\BP{}(\cMin{}(A))$ with $o(n)$-bit auxiliary structures~\cite{DBLP:journals/talg/NavarroS14}, it is sufficient to show how to support $\colorr{}(i)$, $\PRS{}(i)$, and $\NRS{}(i)$ queries using $\BP{}(\cMin{}(A))$ and $c_{min}$. By Lemma~\ref{lem:rrr} and \ref{lem:valid}, we can compute $\colorr{}(i)$ in $O(1)$ time using $o(n)$-bit auxiliary structures by the following procedure: We first check whether the node $i$ is valid using $O(1)$ time by checking the values at the positions $f(i)-1$ and $f(i)-2$ in $\BP{}(\cMin{}(A))$. If the node $i$ is valid (i.e., both the values are $1$), we answer $\colorr{}(i)$ in $O(1)$ time by returning $c_{min}[j]$ where $j$ is $\rank{}_{110}(f(i), \BP(\cMin{}(A)))$; otherwise, by Lemma~\ref{lem:valid}, we answer $\colorr{}(i)$ as blue if and only if the node $i$ is the leftmost child of its parent node. Next, for answering $\PRS{}(i)$ and $\NRS{}(i)$, we construct the following $\ell$-level structure for any positive constant integer $\ell$:

\begin{itemize}[leftmargin=*, listparindent=1.5em]
  \item At the zeroth level, we mark every $(\log n \log \log n)$-th child node and maintain a bit array $M_0[1, \dots, n]$ where $M_0[t] = 1$ if and only if the node $t$ is marked (recall that the node $t$ is the node in $\Min{}(A)$ whose preorder number is $t$). Since there are $n/(\log n \log\log n) = o(n)$ marked nodes, we can store $M_0$ using $o(n)$ bits while supporting $\rank{}$ queries in $O(1)$ time by Lemma~\ref{lem:rrr} (in the rest of the paper, we ignore all floors and ceilings, which do not affect to the results). Also we maintain an array $P_0$ of size $n /(\log n\log\log n)$ where $P_0[j]$ stores both $\PRS{}(s)$ and $\NRS{}(s)$ if $s$ is the $j$-th marked node according to the preorder traversal order. We can store $P_0$ using $O(n \log n /(\log n\log\log n)) = o(n)$ bits.
    
    \item When $\ell > 0$, for the $\ell'$-th level with $0  < \ell' \le \ell$, we mark every $(\log^{(\ell'+1)} n\log^{(\ell'+2)} n)$-th child node. We then maintain a bit array $M_{\ell'}$ which is defined analogously to $M_0$. We can store $M_{\ell'}$  using $o(n)$ bits by Lemma~\ref{lem:rrr}.
    
    Now for any node $p$, let $cr(p)$ be the child rank of $p$, i.e., the number of left siblings of $p$. Also, let $pre_{(\ell'-1)}(p)$ (resp. $next_{(\ell'-1)}(p)$) be the rightmost sibling of $p$ to the left (resp. leftmost sibling of $p$ to the right) which is marked at the $(\ell'-1)$-th level.  

    Suppose $p$ is the $j$-th marked node at the current level in preorder traversal order.
    Then we define an array $P_{\ell'}$ of size $n /(\log^{(\ell')} n\log^{(\ell'+1)} n)$ as $P_{\ell'}[j]$ stores both (i) the smaller value between $cr(p) - cr(\PRS{}(p))$ and $cr(p) - cr(pre_{(\ell'-1)}(p))$, and (ii) the smaller value between $cr(\NRS{}(p))-cr(p)$ and $cr(next_{(\ell'-1)}(p))-cr(p)$. 
    Since both (i) and (ii) are at most $\log^{(\ell')} n \log^{(\ell'+1)} n$, we can store $P_{\ell'}$ using $O(n \log^{(\ell'+1)}n /(\log^{(\ell'+1)} n\log^{(\ell'+2)} n)) = o(n)$ bits. 
    
    Therefore, the overall space is $O(n / \log^{(\ell+2)} n) = o(n)$ bits in total for any positive constant integer $\ell$. 
\end{itemize}

To answer $\PRS{}(i)$ (the procedure for $\NRS{}(i)$ is analogous), we begin by scanning the left siblings of $i$ using the previous sibling operation. Each time a node $i_1$ is visited during the scan, we check in $O(1)$ time whether (i) $\colorr{}(i_1) $ is red, or (ii) $M_{\ell}[i_1] = 1$. If neither condition holds, we continue scanning. By the definition of $M_{\ell}$, one of these conditions is guaranteed to be met within at most $O(\log^{(\ell+1)} n\log^{(\ell+2)} n) = O(\log^{(\ell)} n)$ previous sibling operations, which takes $O(\log^{(\ell)} n)$ time~\cite{DBLP:journals/talg/NavarroS14}. 
If $i_1$ satisfies condition (i), we return $i_1$ as the answer.

If $i_1$ satisfies condition (ii), we jump to $i_1$'s left sibling $i_2$, whose child rank is given by $cr(i_1) - P_{\ell}[j]$, where $j = \rank{}_{1}(i_1, M_{\ell})$. By the definition of $P_{\ell}$, the node $i_2$ always satisfy either $\colorr{}(i_2)$ is red  or $M_{\ell-1}[i_2] = 1$.
We repeat this procedure, decrementing $\ell$ by 1 at each step, until we find a node whose color is red. Since there are at most $\ell$ such iterations and each takes $O(1)$ time, we can answer $\PRS{}(i)$ in $O(\log^{(\ell)} n)$ time in total. We summarize this result in the following theorem.

\begin{theorem}\label{thm:cmin}
Given an array $A[1, \dots, n]$ of size $n$ without consecutive equal elements, we can answer (i) range minimum and previous smaller value queries in $O(1)$ time, and (ii) range $q$-th minimum and next smaller value queries for any $q \ge 1$ in $O(\log^{(\ell)}n)$ time for any positive constant integer $\ell$, using $\BP{}(\cMin{}(A))$ and $c_{min}$ with $o(n)$-bit auxiliary structures.
\end{theorem}

Note that the total space usage of the data structure will be considered in Section~\ref{sec:combine:nonc} (Corollary~\ref{cor:nonconsecutive}). 

\subsection{Combining the encoding data structures on $\cMin{}(A)$ and $\cMax{}(A)$}\label{sec:combine:nonc}
We describe how to combine the data structure of Theorem~\ref{thm:cmin} on $\cMin{}(A)$ and $\cMax{}(A)$ using $3.585n+o(n)$ bits in total. We first briefly introduce the idea of Gawrychowski and Nicholson~\cite{DBLP:conf/icalp/GawrychowskiN15} to combine the DFUDS of $\Min{}(A)$ and $\Max{}(A)$. In DFUDS, any non-root node $i$ is represented as a bit array $0^{d_i}1$ where $d_i$ is the degree of $i$~\cite{DBLP:journals/algorithmica/BenoitDMRRR05}. We refer $0^{d_i}1$ to as the DFUDS of the node $i$.
The encoding of \cite{DBLP:conf/icalp/GawrychowskiN15} is composed of (i) a bit array $U[1, \dots, n]$, where $U[i]$ indicates which tree (either $\Min{}(A)$ or $\Max{}(A)$) has the node $i$ as an internal node, and (ii) a bit array $S = s_1s_2 \dots s_n$, where $s_i$ is the bit array obtained by omitting the first $0$ from the DFUDS of the node $i$ in the tree that $U[i]$ indicates.

To decode the DFUDS of the node $i$ in $\Min{}(A)$ or $\Max{}(A)$, first check whether the tree has the node $i$ as an internal node by referring to $U[i]$. If $i$ is an internal node, one can decode it by prepending $0$ to $s_i$. Otherwise, the decoded sequence is simply $1$ by Lemma~\ref{lem:2dmin}(b). Also, Gawrychowski and Nicholson~\cite{DBLP:conf/icalp/GawrychowskiN15} showed that $U$ and $S$ take at most $3n$ bits in total.
The following lemma shows that a similar idea can also be applied to combine $\BP{}(\cMin{}(A))$ and $\BP{}(\cMax{}(A))$ (the lemma can be proved directly from Lemma~\ref{lem:2dmin}).

\begin{lemma}\label{lem:bpminmax}
For any node $i \in \{1, 2, \dots, n-1\}$, if the node $i$ is an internal node in $\cMin{}(A)$, $f(i+1, \cMin{}(A)) = f(i, \cMin{}(A))+1$, and $f(i+1, \cMax{}(A)) = f(i, \cMax{}(A))+k$, for some $k > 1$. Otherwise, $f(i+1, \cMax{}(A)) = f(i, \cMax{}(A))+1$, and $f(i+1, \cMin{}(A)) = f(i, \cMin{}(A))+k$, for some $k > 1$.
\end{lemma}

\begin{figure}
\begin{center}
\includegraphics[scale=0.7]{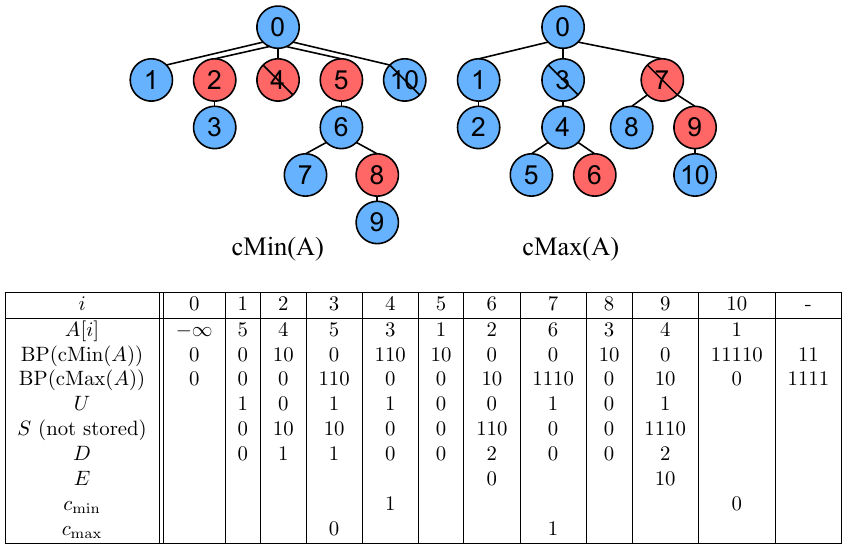}
\end{center}
\caption{Combined data structure of $\cMin{}(A)$ and $\cMax{}(A)$. $i$-th column of the table shows (i) the substring of $\BP{}(\cMin{}(A))$ and $\BP{}(\cMax{}(A))$ begin at position $f(i-1)+1$ and end at position $f(i)$ (shown in the second and the third row, respectively), and (ii) $s_i$ for each $i$ (shown in the fifth row).}
\label{fig:minmaxex1}
\end{figure}

We now describe our combined data structure of $\cMin{}(A)$ and $\cMax{}(A)$. We first maintain the following structures to store $\BP{}(\cMin{}(A))$ and $\BP{}(\cMax{}(A))$:
\begin{enumerate}
    \item 
    A bit array $U[1, \dots, n-1]$ where $U[i] = 0$ (resp. $U[i] = 1$) if the node $i$ is an internal node in $\cMin{}(A)$ (resp. $\cMax{}(A)$). For example, in the Figure~\ref{fig:minmaxex1}, $U[6] = 0$ since the node $6$ is an internal node in $\cMin{}(A)$.
    \item  For each node $i \in \{1, 2, \dots, n-1\}$, suppose the node $i$ is a leaf node in $T \in \{\cMin{}(A), \cMax{}(A)\}$, and let $k_i$ be the number of ones between $f(i, T)$ and $f(i+1, T)$. Now let $S = s_1s_2 \dots s_{n-1}$ be a bit array, where $s_i$  is defined as $1^{k_i-1}0$. For example, in the Figure~\ref{fig:minmaxex1}, since there exist three $1$'s between $f(6)$ and $f(7)$ in $\cMax{}(A)$, $s_6 = 110$. Then, $S$ is well-defined by Lemma~\ref{lem:bpminmax} ($k_i \ge 1$ for all $i$). Also, since there are at most $n-1$ ones and exactly $n-1$ zeros by Lemma~\ref{lem:2dmin}(b), the size of $S$ is at most $2(n-1)$. We maintain $S$ using the following two arrays:
    \begin{enumerate}
        \item An array $D[1, \dots n-1]$ of size $n$ where $D[i] = 0$ if $s_i$ contains no ones, $D[i]=1$ if $s_i$ contains a single one, and $D[i] = 2$ otherwise. For example, in the Figure~\ref{fig:minmaxex1}, $D[6] = 2$, since $s_6$ has two ones. We maintain $D$ using the data structure of Dodis et al.~\cite{DBLP:conf/stoc/DodisPT10}, which can decode any $\Theta(\log n)$ consecutive elements of $D$ in $O(1)$ time using $\ceil{(n-1) \log 3}$ bits. Now let $d_1$ and $d_2$ be the number of $1$'s and $2$'s in $D$, respectively.
        \item Let $i_2$ be the position of the $i$-th $2$ in $D$. Then, we store a bit array $E = e_1e_2, \dots, e_{\ell}$ where $e_i$ is a bit array defined by omitting the first two $1$'s from $s_{i_2}$. For example,in the Figure~\ref{fig:minmaxex1}, since the $6$ is the first position of $D$ whose value is $2$ and $s_6 = 110$, $e_1$ is defined as $0$. The size of $E$ is at most $2(n-1)-(n-1)-(d_1+d_2) = n -d_1-d_2$.
    \end{enumerate}
    \item Store both $f(n, \cMin{}(A))$, and $f(n, \cMax{}(A))$ using $O(\log n)$ bits.
\end{enumerate}
For storing both $c_{min}$ and $c_{max}$, we simply concatenate them into a single array $c_{minmax}$, and store the length of $c_{min}$ using $O(\log n)$ bits. Then, by Lemma~\ref{lem:valid}, the size of $c_{minmax}$ is $d_1+d_2$. 
Thus, our encoding of $\cMin{}(A)$ and $\cMax{}(A)$ takes at most $(n-1) + (n-1) \log 3 + (n-d_1-d_2) + (d_1+d_2) +O(\log n) = (2 + \log 3)n + O(\log n) < 3.585n + O(\log n)$ bits in total~\cite{DBLP:journals/ipl/Tsur19}. An overall example of our encoding is shown in Figure~\ref{fig:minmaxex1}. Now we prove the main theorem in this section.

\begin{theorem}\label{thm:nonconsecutive}
Given an array $A[1, \dots, n]$ of size $n$ without any consecutive equal elements, there exists a $(3.585n+o(n))$-bit encoding data structure which can answer (i) range minimum/maximum and previous larger/smaller value queries in $O(1)$ time, and (ii) range $q$-th minimum/maximum and next larger/smaller value queries in $O(\log^{(\ell)} n)$ time, for any $q \ge 1$ and positive constant integer $\ell$.
\begin{proof}
We show how to decode any $\log n$ consecutive bits of $\BP(\cMin{}(A))$ in $O(1)$ time, which proves the theorem by Lemma~\ref{lem:rrr} and Theorem~\ref{thm:cmin}. 
Note that the auxiliary structures and the procedure for decoding $\BP(\cMax{}(A))$ are analogous. 
Let $B[1, \dots, f(n)-1]$ be a subarray of $\BP(\cMin{}(A))$ of size $f(n)-1$, which is defined as $\BP(\cMin{}(A))[2, \dots, f(n)]$. Then it is enough to show how to decode $\log n$ consecutive bits of $B$ in $O(1)$ time using $o(n)$-bit auxiliary structures (note that $\BP(\cMin{}(A))$ is $0 \cdot B \cdot 1^{2n+2-f(n)}$). We also denote $f(n)-1$ by $f'(n)$ in this proof.
For each position $j \in \{1, \dots, f'(n)\}$ of $B$, we first define a corresponding position of $j$ in $U$, $D$, and $E$, denoted as $\alpha(j)$, $\beta(j)$, and $\gamma(j)$ respectively.
Briefly, $\alpha(j)$, $\beta(j)$, and $\gamma(j)$ denote the leftmost positions in $U$, $D$, and $E$, respectively, that are needed to decode the suffix $B[j, \dots, f(n)-1]$. 
We define $\alpha(j)$, $\beta(j)$, and $\gamma(j)$ more precisely as follows:

\begin{enumerate}
    \item $\alpha(j) = \rank{}_{0}(j, B)$
    \item $\beta(1) =  \gamma(1) = 1$
    \item For $j > 1$, $\beta(j) = \rank{}_{0}(j-1, B)$. 
    \item For $j > 1$, let $j_2$ be the number of $2$'s in $D[1, \dots, \beta(j)]$ and 
    $j_1$ be the number of $1$'s in $B$ between $B[j]$ and the leftmost $0$ to the right.
    Then $\gamma(j)$ is defined as (i) $1$ if $j_2$ is $0$, (ii) $\select{}_0(j_2, E)+1$ if $j_2 > 0$ and $D[\beta(j)] \neq 2$, and (iii) $\select{}_0(j_2, E)-\max{}(j_1-3, 0)$ otherwise. 
\end{enumerate}
In the rest of the proof, we describe how to decode a subarray of $B$ starting from the position $j$ from subarrays of $U$, $D$ and $E$ starting from the positions $\alpha(j)$, $\beta(j)$ and $\gamma(j)$, respectively.

For $i \in \{1, 2, \dots, \ceil{(f'(n))/\log n}\}$, we define the $i$-th block of $B$ as $B[\ceil{(i-1)\log n}+1, \dots, \min{}(\ceil{i \log n}, f'(n))]$. Then, it is enough to decode at most two consecutive blocks of $B$ to decode any $\log n$ consecutive bits of $B$. 
Next, we define the $i$-th block of $U$, $D$, and $E$ as follows:

\begin{itemize}
\item The $i$-th block of $U$ is defined as a subarray of $U$ starting at position $\alpha(\ceil{(i-1)\log n} + 1)$ and ending at position $\alpha(\min(\ceil{i \log n}, f'(n)))$. To decode the blocks of $U$ without $B$, we mark all the starting positions of the blocks of $U$ using a bit array $U_1$ of size $f'(n)$ where $U_1[i]= 1$ if and only if the position $i$ is the starting position of the block in $U$. Then, since $U_1$ contains at most $O(f'(n)/\log n) = o(n)$ $1$'s, we can store $U_1$ using $o(n)$ bits while supporting $\rank{}$ and $\select$ queries in $O(1)$ time by Lemma~\ref{lem:rrr}.

Also, to handle the case where two distinct blocks of $U$ share the same starting position, we define an additional bit array $U_2$ of size $\ceil{f'(n)/\log n}$, where $U_2[i] = 1$ if and only if the $i$-th block of $U$ has the same starting position as the $(i - 1)$-th block.
We store $U_2$ using the data structure from Lemma~\ref{lem:rrr}, which supports $\rank$ and $\select$ queries in $O(1)$ time using $o(n)$ bits. Then any block of $U$ can be decoded in $O(1)$ time using $\rank$ and $\select$ queries on $U_1$ and $U_2$, as each block has size at most $\log n$; this follows from the fact that every position in $U$ corresponds to at least one position in $B$.

\item The $i$-th block of $D$ is defined as the subarray of $D$ starting at position $\beta(\ceil{(i-1)\log n} + 1)$ and ending at position $\beta(\min(\ceil{i \log n}, f'(n)))$.
To decode the blocks, we maintain two bit arrays $D_1$ and $D_2$, analogous to $U_1$ and $U_2$, respectively, using a total of $o(n)$ bits.
With these, each block of $D$ can be decoded in $O(1)$ time using $\rank$ and $\select$ operations on $D_1$ and $D_2$ from the same argument as for decoding the the blocks of $U$.

\item The $i$-th block of $E$ is defined as a subarray of $E$ whose starting and ending positions are $\gamma(\ceil{(i-1)\log n}+1)$ and $\gamma(\min{}(\ceil{i \log n}, f'(n)))$, respectively. To decode the blocks of $E$, we maintain two bit arrays $E_1$ and $E_2$ analogous to $U_1$ and $U_2$, respectively, using $o(n)$ bits. Note that, unlike $U$ or $D$, the size of some blocks in $E$ can be arbitrarily large, since some positions in $E$ may not have corresponding positions in $B$—specifically when the runs of $1$ originate from $\BP{}(\cMax{}(A))$.
 
To handle this case, we classify each block of $E$ as \textit{bad block} and $\textit{good block}$, where the size of bad block is at least at $c \log n$ for some constant $c \ge 9$, whereas the size of good block is less than $c \log n$. If the $i$-th block of $E$ is good (resp. bad), we say it as $i$-th good (resp. bad) block. Then, each $i$-good block of $E$ can be decoded in $O(1)$ time by accessing at most $c \log n$ consecutive bits of $E$ starting from the leftmost position of the block.

Next, for each $i$-th bad block of $E$, let $F_i$ be a subsequence of the $i$-th bad block, which consists of all bits at the position $j$ where $\gamma^{-1}(j)$ exists.
We store $F_i$ explicitly using $\log n$ bits, which takes $\Theta(n)$ bits in total. 
However, we can apply the same argument used in \cite{DBLP:conf/icalp/GawrychowskiN15} as follows:
Briefly, the argument states that any two bad blocks of $E$—one defined from $\BP{}(\cMin{}(A))$ and the other from $\BP{}(\cMax{}(A))$—can share at most $4 \log n$ bits, since each position in $E$ corresponds to at least one position in either $\BP{}(\cMin{}(A))$ or $\BP{}(\cMax{}(A))$.
Therefore, for each $i$-th bad block of $E$ (defined from $\BP{}(\cMin{}(A))$), we can save at least $\log n$ bits by maintaining the remaining $(c - 4) \log n$ bits of $E$ as a compressed from (using Lemma~\ref{lem:rrr}), since it contains at most $\log n$ zeros.
As a result, we can maintain $F_i$ for all $i$-th bad blocks of $E$ without increasing the total space usage.
(see the proof of Theorem 1 in \cite{DBLP:conf/icalp/GawrychowskiN15} for a detailed argument).
\end{itemize}

Next, let $g(u, d, e, b)$ be a function, which returns a subarray of $B$ from the subarrays of $U$, and $D$, and $E$ as follows (assuming $u = u[1] \cdot u'$ and $d = d[1] \cdot d'$): 

\[
g(u, d, e, b)= 
\begin{cases}
\epsilon &\text{if}~u = \epsilon~\text{or}~d = \epsilon \\
0 \cdot g(u', d', e, b)& \text{if}~u[1] \neq b~\text{and}~d[1] \neq 2\\
0 \cdot g(u', d', e', b)& \text{if}~u[1] \neq b,~d[1]=2,~\text{and}~e = 1^t0 \cdot e'\\
10 \cdot g(u', d', e, b)      & \text{if}~u[1]=b~\text{and}~d[1]=0\\
110 \cdot g(u', d', e, b)      & \text{if}~u[1]=b~\text{and}~d[1]=1\\
1^{t+3}0 \cdot g(u', d', e', b)      & \text{if}~u[1]=b,~d[1]=2,~\text{and}~e = 1^t0 \cdot e'\\
\end{cases}
\]

We store a precomputed table that stores $g(u, d, e, b)$ for all possible $u$, $d$, and $e$ of sizes $\frac{1}{4}\log n$ and $b \in \{0, 1\}$ using $O(2^{\frac{1}{4} \log n + \frac{3}{2} \cdot \frac{1}{4} \log n + \frac{1}{4} \log n } \log n) = O(n^{\frac{7}{8}} \log n) = o(n)$ bits. Here, $b$ is fixed to $0$ (resp. $1$) when decoding a subarray of $\BP(\cMin{}(A))$ (resp. $\BP(\cMax{}(A))$). 
Finally, note that there are at most $q \le 4$ consecutive positions from $p$ to $p+q-1$ of $B$ whose corresponding positions are the same in both $D$ and $E$. Because such a case can only occur when $B[p] = B[p+1] = \dots = B[p+q-2] = 1$ and $B[p+q-1] = 0$, we maintain an array $R$ of size $O(n/\log n) = o(n)$, which stores the four cases of the number of consecutive $1$'s ($0$, $1$, $2$, or at least $3$) from the beginning of the $i$-th block of $B$

To decode the $i$-th block of $B$, we first decode the $i$-block of $U$ and $D$ in $O(1)$ time. Let these subarrays be $b_u$ and $b_d$, respectively. We define $b_e$ as $F_i$ if the $i$-th block of $E$ is bad; otherwise, we define $b_e$ as the $i$-th good block of $E$. In either case, $b_e$ can be decoded in $O(1)$ time. 
Next, we compute $g(b_u, b_d, b_e, 0)$ in $O(1)$ time by accessing the precomputed table $O(1)$ times, and prepend a $0$ if we are decoding the first block of $B$. Finally, if the number of consecutive $1$s at the beginning of $g(b_u, b_d, b_e, 0)$ is at most $3$, we remove some of these $1$s by referring to $R$.
\end{proof}
\end{theorem}

In order to only support the queries on $\cMin{}(A)$, the data structure described in Theorem~\ref{thm:nonconsecutive} can be modified by performing the following steps: (i) eliminating $U$, and (ii) replacing $S$ with $\BP{}(\cMin{}(A))$ (i.e., instead of $S$, representing $\BP{}(\cMin{}(A))$ using $D$ and $E$). 
As a result of these modifications, the sizes of $E$ and $c_{\min}$ become $n-1-d_1-d_2$ and $d_2$, respectively (recall that $d_1$ and $d_2$ are the number of $1$'s and $2$'s in $D$, respectively). Consequently, the total space usage is calculated as $(n-1)\log 3 + (n-1-d_1) + O(\log n) \le 2.585n-d_1 + O(\log n)$ bits.
We summarize the result in the following Corollary.

\begin{corollary}\label{cor:nonconsecutive}
Given an array $A[1, \dots, n]$ of size $n$ and any positive constant integer $\ell$, suppose $A$ has no two consecutive equal elements. Then there exists a $(2.585n-d_1+o(n))$-bit encoding data structure which can answer (i) range minimum and previous smaller value queries in $O(1)$ time, and (ii) range $q$-th minimum and next smaller value queries in $O(\log^{(\ell)} n)$ time, for any $q \ge 1$. Here, $d_1$ denotes the number of positions $i \in \{2, \dots, n\}$ in $A$ which satisfy $\PSV{}(i-1) = \PSV{}(i)$, i.e., the number of nodes in $\cMin{}(A)$ that have a leaf node as their immediate left sibling.
\end{corollary}


\section{Data structure on general arrays}\label{sec:equal}
In this section, we present a  $(3.701n + o(n))$-bit data structure to support the range $q$-th minimum/maximum and next/previous larger/smaller value queries on the array $A[1 \dots, n]$ without any restriction. Let $C[1, \dots, n]$ be a bit array of size $n$ where $C[1] = 0$, and for any $i >1$, $C[i]=1$ if and only if $A[i-1] = A[i]$. If $C$ has $k$ ones, we define an array $A'[1, \dots, n-k]$ of size $n-k$ that discards all consecutive equal elements from $A$. 
Then, by the definition of the colored 2d-min and 2d-max heaps, we observe that if $C[i] = 1$, then node $i$ is a blue-colored node in both $\cMin(A)$ and $\cMax(A)$, and its immediate left sibling is a leaf node in both trees.
Furthermore, by deleting all the bits at the positions $f(i, \cMin{}(A))-1$, and $f(i, \cMin{}(A))$ from $\BP{}(\cMin{}(A))$ we can obtain $\BP{}(\cMin{}(A'))$. We can also obtain $\BP{}(\cMax{}(A'))$ from $\BP{}(\cMax{}(A'))$ analogously. Now we prove the following theorem.

\begin{figure}
\begin{center}
\includegraphics[scale=0.7]{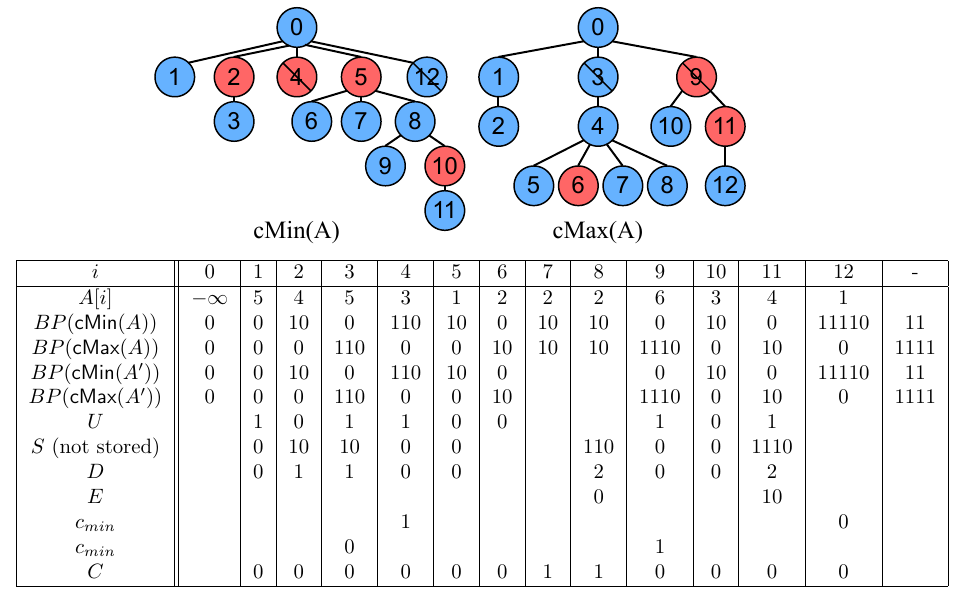}
\end{center}
\caption{Combined data structure of $\cMin{}(A)$ and $\cMax{}(A)$. Note that $A'$ is the same array as the array in Figure~\ref{fig:minmaxex1}.}
\label{fig:minmaxex2}
\end{figure}

\begin{theorem}\label{thm:general}
Given an array $A[1, \dots, n]$ of size $n$ and any positive constant integer $\ell$, there exists a $(3.701n+o(n))$-bit encoding data structure which can answer (i) range minimum/maximum and previous larger/smaller value queries in $O(1)$ time, and (ii) range $q$-th minimum/maximum and next larger/smaller value queries in $O(\log^{(\ell)} n)$ time, for any $q \ge 1$.
\begin{proof}
The data structure consists of $C$ and the data structure of Theorem~\ref{thm:nonconsecutive} on $A'$, which can answer all the queries on $A'$ in $O(\log^{(\ell)} n)$ time
(see Figure~\ref{fig:minmaxex2} for an example).
By maintaining $C$ using the data structure of Lemma~\ref{lem:rrr}, the data structure takes at most $(2+\log 3)(n-k) + {n \choose k} +o(n) \le 3.701n+o(n)$ bits in total~\cite{DBLP:journals/ipl/Tsur19} while supporting $\rank{}$ and $\select{}$ queries on $C$ in $O(1)$ time. For any node $i$ in $\cMin{}(A)$ and $\cMax{}(A)$, we can compute the color of the node $i$ in $O(1)$ time as follows. If $C[i] = 0$, we return the color of the node $(\rank{}_{0}(i, C)-1)$ in $\cMin{}(A')$ and $\cMax{}(A')$, respectively. Otherwise, we return blue. Now we describe how to decode any $\log n$ consecutive bits of $\BP(\cMin{}(A))$ in $O(1)$ time using $o(n)$-bit auxiliary structures, which proves the theorem (the auxiliary structures and the procedure for decoding $\BP(\cMax{}(A))$ are analogous). In the proof, we denote $\BP(\cMin{}(A))$ and $\BP(\cMin{}(A'))$ as $B_{A}$ and $B_{A'}$, respectively.

For each position $j$ of $B_{A}$, we say $j$ is $\textit{original}$ if $B_{A}[j]$ comes from the bit in $B_{A'}$, and $\textit{additional}$ otherwise. That is, the position $j$ is additional if and only if $j \in \{f(j')-1, f(j') \mid C[j'] =1\}$.
For each original position $j$, let $b'(j)$ be its corresponding position in $B_{A'}$.

Now we divide $B_{A}$ into the blocks of size $\log n$ except the last block, and let $s_i$ be the starting position of the $i$-th block of $B_{A}$. We then define a bit array $M_{B}$ of size $2(n-k)$ as follows. For each $i \in \{1, \dots, \ceil{(2(n+1)/\log n}\}$, we set the $b'(s_i)$-th position of $M_{B}$ as one if $s_i$ is original. Otherwise, we set the $b'(s'_i)$-th position of $M_{B}$ as one where $s'_i$ is the leftmost original position from $s_i$ to the right in $B_{A}$. All other bits in $M_{B}$ are $0$.
Also, let $M'_{B}$ a bit array of size $\ceil{(2(n+1)/\log n}$  where $M'_{B}[i]$ is $1$ if and only if we mark the same position for $s_i$ and $s_{i-1}$. Since $M_{B}$ has at most $\ceil{(2(n+1)/\log n} = o(n)$ ones, we can maintain both $M_{B}$ and $M'_{B}$ in $o(n)$ bits while supporting $\rank{}$ and $\select{}$ queries in $O(1)$ time by Lemma~\ref{lem:rrr}. Similarly, we define a bit array $M_{C}$ of size $n$ as follows. If $s_i$ is original, we set the $(\rank{}_{0}(s_i-1, B))$-th position of $M_{C}$ as one. Otherwise, we set the $(\rank{}_{0}(s_i-1, B))$-th (resp. $(\rank{}_{0}(s_i, B)$-th) position of $M_{C}$ as one if $B_{A}[s_i]$ is $0$ (resp. $1$). We also maintain a bit array $M'_{C}$ analogous to $M'_{B}$. Again, we can maintain both $M_{C}$ and $M'_{C}$ using $o(n)$ bits while supporting $\rank{}$ and $\select{}$ queries on them in $O(1)$ time.

Next, let $h(b, c)$ be a function, which returns a subarray of $B_{A}$ from the subarrays of $B_{A'}$ and $C$, defined as follows (assuming $c = c[1] \cdot c'$):

\[
h(b,c)=
\begin{cases}
1^t & \text{if}~b = 1^t~\text{and}~c[1] = 0\\
1^t0 \cdot h(b', c')& \text{if}~b = 1^t0 \cdot b'~\text{and}~c[1] = 0\\
10 \cdot h(b, c')& \text{if}~c[1] = 1\\
\end{cases}
\]

We store a precomputed table, which stores $h(b, c)$ for all possible $b$, $c$ of size $\frac{1}{4}\log n$ using $O(2^{\frac{1}{2}\log n} \log n) = O(\sqrt{n} \log n) = o(n)$ bits. Finally, we store a bit array $M_S$ of size $o(n)$, which indicates whether the first bit of the $i$-th block of $B_{A}$ is $0$ or not.

To decode the $i$-th block of $B_{A}$, we first decode $\log n$-sized subarrays of $B_{A'}$ and $C$, $b_{b'}$ and $b_c$, whose starting positions are $\select{}_{1}(\rank{}_{0}(i, M'_{B}), M_{B})$ and $\select{}_{1}(\rank{}_{0}(i, M'_{C}), M_{C})$, respectively. We then compute $h(b_{b'}, b_{c})$ in $O(1)$ time by referring to the precomputed table $O(1)$ times. As the final step, we remove the leftmost bit of $h(b_{b'}, b_{c})$ if the $i$-th block of $B_{A}$ starts from $0$, and $s_i$ is additional (this can be done by referring $M_S$).
\end{proof}
\end{theorem}

Similar to the data structure presented in Corollary~\ref{cor:nonconsecutive}, the data structure of Theorem~\ref{thm:general} can be modified to support queries only for $\cMin(A)$ when $A$ contains consecutive equal elements. We provide a summary of this result in the following Corollary.

\begin{corollary}\label{cor:general}
Given an array $A[1, \dots, n]$ of size $n$ and any positive constant integer $\ell$, let $A'$ be an array that discards all consecutive equal elements from $A$. 
Then there exists a $(2.808n-d_1+o(n))$-bit encoding data structure which can answer (i) range minimum and previous smaller value queries in $O(1)$ time, and (ii) range $q$-th minimum and next smaller value queries in $O(\log^{(\ell)} n)$ time, for any $q \ge 1$. Here, $d_1$ denotes the number of positions $i \in \{2, \dots, |A'|\}$ in $A'$ which satisfy $\PSV{}(i-1) = \PSV{}(i)$, i.e., the number of nodes in $\cMin{}(A')$ that have a leaf node as their immediate left sibling.
\begin{proof}
In order to only support the queries on $\cMin{}(A)$, 
we can combine the data structure of Corollary~\ref{cor:nonconsecutive} on $A'$ with $C$ using $2.585(n-k) + {n \choose k}- d_1 +o(n) \le 2.808n-d_1+o(n)$ bits~\cite{DBLP:journals/ipl/Tsur19}. The decoding procedure is the same as in the proof of Theorem~\ref{thm:general}. 
\end{proof}
\end{corollary}

Interestingly, the following theorem shows that even when $A$ contains consecutive equal elements, there exists a data structure of size at most $2.585n+o(n)$ bits that can support all the queries on $\cMin{}(A)$ with the same asymptotic query time as the data structure of Corollary~\ref{cor:general}. 

\begin{figure}
\begin{center}
\includegraphics[scale=0.8]{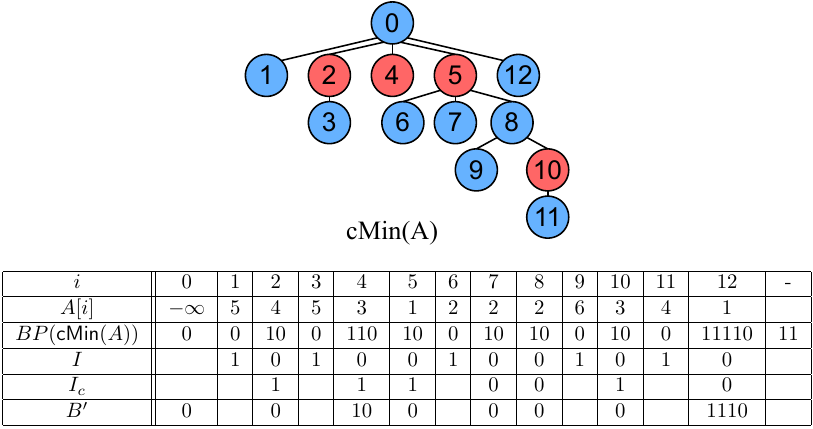}
\end{center}
\caption{An example of the data structure of Theorem~\ref{thm:cmin_second} on the array $A$, which is the same array as in Figure~\ref{fig:minmaxex2}.}
\label{fig:thm4ex}
\end{figure}

\begin{theorem}\label{thm:cmin_second}
Given an array $A[1, \dots, n]$ of size $n$, there exists a data structure using at most $(2.585n+o(n))$ bits of space that can answer (i) range minimum and previous smaller value queries in $O(1)$ time, and (ii) range $q$-th minimum and next smaller value queries for any $q \ge 1$ in $O(\log^{(\ell)}n)$ time for any positive constant integer $\ell$.
\begin{proof}
As in the proof of Theorem~\ref{thm:general}, we use the notation $B_A$ to denote $\BP(\cMin(A))$. Then it suffices to show that we can decode any $\Theta(\log n)$ consecutive bits of $B_{A}[1, \dots, f(n)]$ and $\colorr(i)$ in $O(1)$ time using a data structure of size at most $2.585n + o(n)$ bits. 
We maintain two bit arrays: $I$ of size $n$ and $I_c$ of size $n - m$, where $m$ denotes the number of leftmost children in $\cMin(A)$. 
These arrays are defined as follows: (i) for each $i \in {1, \dots, n}$, $I[i] = 1$ if and only if node $i$ is the leftmost child in $\cMin(A)$; and 
(ii) for each $i \in \{1, \dots, n - m\}$, $I_c[i] = 1$ if and only if the color of the node
$\select_0(i, I)$ is red.
As $I$ contains exactly $m$ ones, both $I$ and $I_c$ can be stored using at most $\log {{n \choose m}} + (n-m) +o(n) \le 1.585n + o(n)$ bits~\cite{DBLP:journals/ipl/Tsur19} using the data structure presented in Lemma~\ref{lem:rrr}, which supports $\rank$ and $\select$ queries on $I$ in $O(1)$ time. Using the data structures, we can compute $\colorr(i)$ in $O(1)$ time by returning red if and only if $I[i] = 0$ and $I_c[\rank_{0}(i, I)] = 1$. Note that the color of the leftmost children is always blue, even when $A$ has  consecutive equal elements.

Next, suppose $B_{A}$ is given as $b_0b_1b_2 \dots b_n$, where $b_0 = 0$ and each $b_i$ corresponds to the substring $B_{A}[f(i-1)+1, \dots, f(i)]$ for $i \in \{1, \dots, n\}$. We define a bit array $B' = b'_0b'_1 \dots b'_n$ as follows: (i) set $b'_0 = 0$; (ii) if $f(i-1)+1 = f(i)$ (i.e., node $i$ is the leftmost child of node $i-1$ in $\cMin(A)$), then $b'_i$ is an empty array; (iii) otherwise, $b'_i$ is obtained by removing the first bit of $b_i$.

From the construction, either $B_{A}[f(i)]$ or $B_{A}[f(i)-1]$ is removed from $B_{A}$ when constructing $B'$ for all $i \in \{1, \dots, n\}$. Thus, the size of $B'$ is at most $n+2$ bits. As a result, total size of $I$, $I_c$, and $B'$ is at most $2.585n + o(n)$ bits (see Figure~\ref{fig:thm4ex} for an example). 
Also for each position $p_b$ of $B_{A}$, we define its corresponding position $p_{b'}$ of $B'$ as follows:
\begin{itemize}
    \item $p_{b'}$ is $1$ if $p_b = 1$, and
    \item If $B_{A}[p_b]$ is removed during the construction of $B'$, $p_{b'}$ is the leftmost position to the right of $p_b$ where $B_{A}[p_{b'}]$ is not removed during the construction, and
    \item Otherwise, if $p_b$ is the $j$-th position of $b_i$ of $B_{A}$, $p_{b'}$ is the $j$-th (if $B_{A}[p_b] = 1$) or $(j-1)$-th position (if $B_{A}[p_b] = 0$) of $b'_i$ in $B'$.  
\end{itemize}

For $j \in \{1, \dots, \ceil{n/\log n}\}$, we define the $j$-th block of $B_{A}$ as a subarray $B_{A}[\ceil{(j-1)\log n}+1, \dots, \min{}(\ceil{j \log n}, n)]$. Now we describe how to decode the $j$-th block of $B_{A}$ in $O(1)$ time using $I$ and $B'$, which prove the theorem. Let $s_j = \ceil{(h-1)\log n}+1$ be the starting position of the $j$-th block of $B_{A}$. We then maintain the following bit arrays:

\begin{enumerate}
\item $N_B$: a bit-array of size $|B'|$ where $N_B[j]$ = 1 if and only if $B'[j]$ corresponds to the starting position of any block in $B_{A}$.
\item $N'_B$: a bit-array of size $\ceil{n/\log n}$ where $N'_B[j]$ = 1 if and only if $j > 1$ and both the positions $s_{j}$ and $s_{j-1}$ of $B_{A}$ have the same corresponding positions in $B'$.
\item $N_I$: a bit-array of size $n$ constructed as follows. If $B_{A}[s_j]=0$ (resp. $1$), then $N_I[\text{rank}_0(s_j, B)]$ (resp. $N_I[\text{rank}_0(s_j, B)+1]$) is $1$. The remaining bits in $N_I$ are set to $0$.  
\item $N'_I$: a bit-array of size $\ceil{n/\log n}$ where $N'_B[j]$ = 1 if and only if $j > 1$, and $B_{A}[s_j]  = B[s_{j-1}]$, and $\rank_{0}(s_j, B) = \rank_{0}(s_{j-1}, B)$.
\item $B_S$: a bit array of size $\ceil{n/\log n}$ where $B_S[j]$ is $B_{A}[s_j]$. 
\end{enumerate}

Since both $N_B$ and $N_I$ contain at most $\ceil{n/\log n} = o(n)$ ones, it is possible to store all the bit-arrays mentioned above using a total of $o(n)$ bits while supporting $O(1)$-time $\text{rank}$ and $\text{select}$ queries. This can be achieved by using the data structure presented in Lemma~\ref{lem:rrr}. 
In addition, let $g_2(b', I_b)$ be a function that returns a subarray of $B_{A}$ from the two subarrays $b'$ and $I_b$ of $B'$ and $I$, respectively as follows (we assume $b'=b'[1]\cdot b''$ and $I_b=I_b[1]\cdot I'_b$):

\[
g_2(b', I_b)=
\begin{cases}
0 \cdot g_2(b', I'_b)& \text{if}~I_b[1] = 1\\
1 \cdot g_2(b'', I_b)& \text{if}~I_b[1] = 0~\text{and}~b'[1] = 1\\
10 \cdot g_2(b'', I'_b)& \text{if}~I'_b[1] = 0~\text{and}~b'[1] = 0\\
\end{cases}
\]

We store a precomputed table, which stores $f(b', I_b)$ for all possible $b'$ and $I_b$ of size $\frac{1}{4}\log n$ using $O(2^{\frac{1}{2}\log n} \log n) = O(\sqrt{n} \log n) = o(n)$ bits. Consequently, the total size of auxiliary structures is $o(n)$ bits. 

In order to decode the $j$-th block of $B_{A}$, we follow these steps: Firstly, we decode a substring $sub_b$ of $B'$ with a length of $\ceil{\log n}$, starting from the position $\select_{1}(\rank_{0}(j, N'_B), N_B)$. We also decode a substring $sub_I$ of $I$, with a length of $\ceil{\log n}$, starting from the position $\select_{1}(\rank_{0}(j, N'_I), N_I)$. Next, we compute a prefix of $g_2(sub_b, sub_I)$ that has a size of $\ceil{\log n}$ in $O(1)$ time using the precomputed table. Finally, we remove the first $1$ from the prefix if the first bit of $g_2(sub_b, sub_I)$ is $1$ and $B_S[j]$ is $0$.
\end{proof}
\end{theorem}

\begin{figure}[ht]
    \begin{center}
        \includegraphics[scale=0.28]{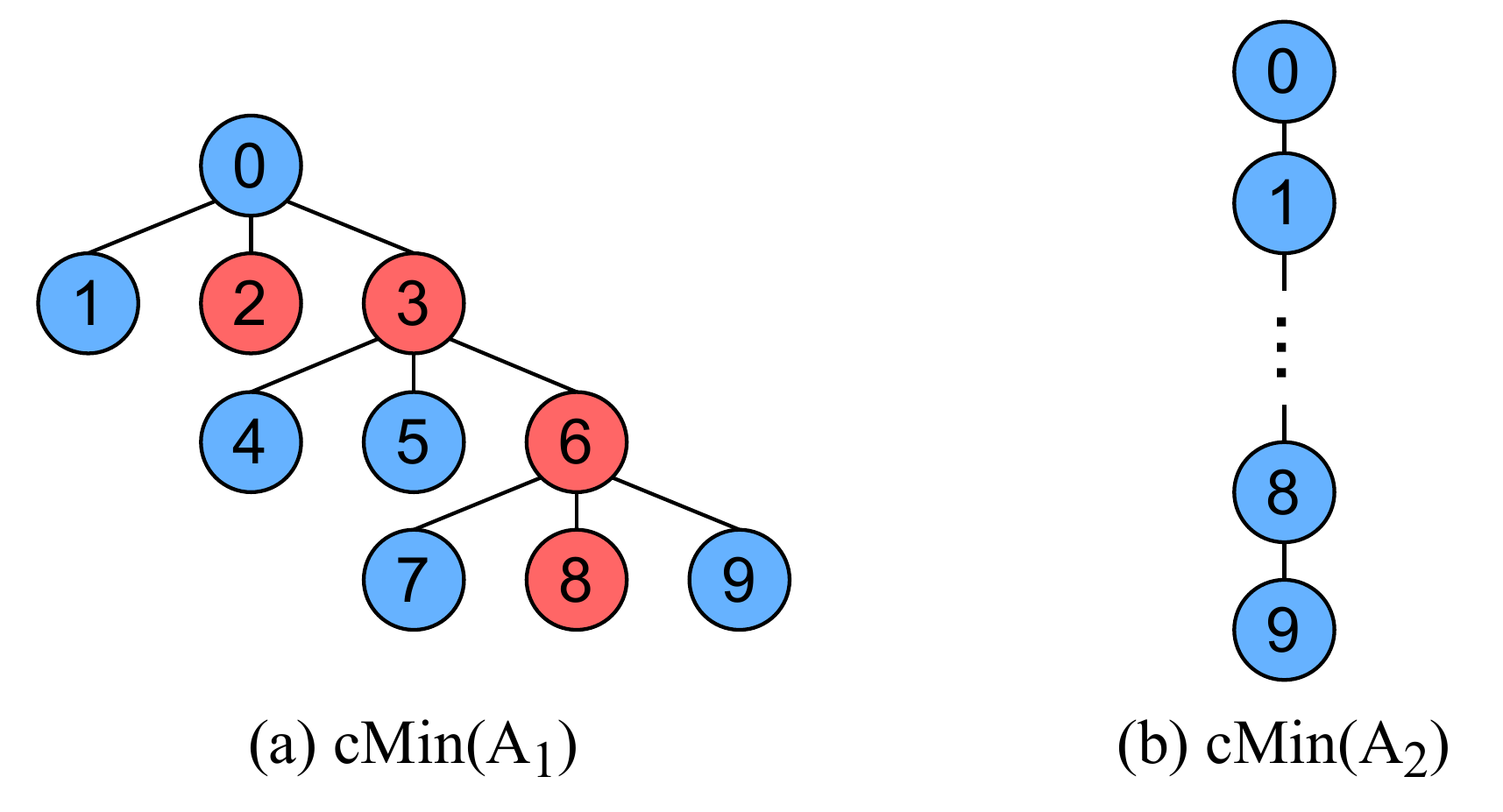}
    \end{center}
    \caption{$\cMin{}(A_1)$ and $\cMin{}(A_2)$ on the array $A_1 =~3~2~1~5~5~4~7~6~6$ and $A_2 =~1~2~3~4~5~6~7~8~9$. 
    }
    \label{fig:example1}
\end{figure}
\begin{remark}
Comparing the data structures presented in Corollary~\ref{cor:general} and Theorem~\ref{thm:cmin_second}, we observe that the latter requires less space in the worst-case. However, if the array contains a long decreasing sequence, resulting in a large value of $d_1$, the data structure of Corollary~\ref{cor:general} may use less space than that of Theorem~\ref{thm:cmin_second}. 

For example, consider the array $A_1$ in Figure~\ref{fig:example1}(a) of size $n = 9$. In this case, $d_1 = 2n/3$, and $\cMin(A_1)$ contains $n/3$ leftmost children. Hence, the data structure from Theorem~\ref{thm:cmin_second} uses $\log {{n} \choose {n/3}} + 2n/3 + n + o(n) \le 2.585n + o(n)$ bits~\cite{DBLP:journals/ipl/Tsur19}, while the data structure from Corollary~\ref{cor:general} uses $2.808n - \frac{2}{3}n + o(n) \le 2.142n + o(n)$ bits.

On the other hand, if the input array $A_2$ of size $n$ is strictly increasing, then $\PSV(i-1) \neq \PSV(i)$ holds for all $i \in \{2, \dots, n\}$ (implying $d_1 = 0$), and $\cMin{}(A_2)$ consists of $n$ leftmost children (see Figure~\ref{fig:example1}(b)). In this case, the data structure from Corollary~\ref{cor:nonconsecutive} (i.e., Corollary~\ref{cor:general} for arrays without consecutive equal elements) uses $2.585n + o(n)$ bits, whereas the data structure from Theorem~\ref{thm:cmin_second} uses only $n + o(n)$ bits. 

Note that these examples also show that, for certain input instances, the data structures of Corollary~\ref{cor:general} (Corollary~\ref{cor:nonconsecutive} for arrays without consecutive equal elements) and Theorem~\ref{thm:cmin_second} can use less space than Fischer’s worst-case succinct data structure, which uses $2.54n + o(n)$ bits ($2n + o(n)$ bits for arrays without consecutive equal elements)~\cite{DBLP:journals/tcs/Fischer11}.
\end{remark}

\section{Lower bounds}\label{sec:lower}
This section considers the effective entropy to answer range $q$-th minimum and maximum queries on an array of size $n$, for any $q \ge 1$. Note that for any $i \in \{1, \dots, n\}$, both $\PSV{}(i)$ and $\PLV{}(i)$ queries can be answered by computing range $q$-th minimum and maximum queries on the suffixes of the substring $A[1, \dots, i]$, respectively. Similarly, both $\NSV{}(i)$ and $\NLV{}(i)$ queries can be answered by computing range $q$-th minimum and maximum queries on the prefixes of the substring $A[i, \dots, n]$, respectively.

Let $\mathcal{A}_n$ be a set of all arrays of size $n \ge 2$ constructed from the following procedure:

\begin{enumerate}
    \item For any $0 \le k \le n - 1$, choose any $k$ positions from the set $\{2, \dots, n\}$, and let $X$ be a set of the remaining $n - k$ unselected positions.
    We then construct a \textit{Baxter permutation}~\cite{10.2307/2034894} $\pi_{n-k} : X \rightarrow \{1, \dots, n - k\}$, which is a permutation of size $n - k$ over the positions in $X$.
    
    A Baxter permutation is one that avoids any three indices $j_1 < j_2 < j_3$ satisfying either $\pi_{n-k}(j_2+1) < \pi_{n-k}(j_1) < \pi_{n-k}(j_3) < \pi_{n-k}(j_2)$ or $\pi_{n-k}(j_2) < \pi_{n-k}(j_3) < \pi_{n-k}(j_1) < \pi_{n-k}(j_2+1)$—that is, $\pi_{n-k}$ avoids the patterns $2-41-3$ and $3-14-2$. For example, $\pi_{n-k} = 3~5~6~1~4~2$ is not a Baxter permutation since it contains the $2-41-3$ pattern: $\pi_{n-k}(4) < \pi_{n-k}(1) < \pi_{n-k}(5)  < \pi_{n-k}(3)$. 
    
    \item For each of the $k$ selected positions $j$, assign the value $\pi_{n-k}(j')$, where $j'$ is the immediate predecessor of $j$ in $X$. This is well-defined since $X$ always includes the leftmost position in the array. Moreover, the construction ensures that equal elements in the array appear in consecutive positions.
\end{enumerate}

Since the number of all possible Baxter permutations of size $n-k$ is at most $2^{3(n-k)-\Theta(\log n)}$~\cite{DBLP:conf/icalp/GawrychowskiN15}, the effective entropy of $\mathcal{A}_n$ is at least $\log |\mathcal{A}_n| \ge \log (\sum_{k=0}^{n-1} 2^{3(n-k)-\Theta(\log n)} \cdot {n-1 \choose k} ) \ge \max{}_{k}(3n-3k+\log{n \choose k} - \Theta(\log n)) \ge n\log 9 -\Theta(\log n) \ge 3.16n - \Theta(\log n)$ bits~\cite{DBLP:journals/ipl/Tsur19}. The following theorem shows that the effective entropy of the encoding to support the range $q$-th minimum and maximum queries on an array of size $n$ is at least $3.16n-\Theta(\log n)$ bits.

\begin{theorem}\label{thm:lb}
Any array $A$ in $\mathcal{A}_n$ for $n \ge 2$ can be reconstructed using range $q$-th minimum and maximum queries on $A$. 
\begin{proof}
We follow the same argument used in the proof of Lemma 3 in~\cite{DBLP:conf/icalp/GawrychowskiN15}, which shows that one can reconstruct any Baxter permutation of size $n$ using range minimum and maximum queries.

The proof is induction on $n$. the case $n = 2$ is trivial since only the possible cases are $\{1, 1\}$ or $\{1, 2 \}$, which can be decoded by range first and second minimum queries.
Now suppose the theorem statement holds for any size less than $n \ge 3$. Then, both $A_1 = A[1, \dots, n-1]$ and $A_2 = A[2, \dots, n]$ from $\mathcal{A}_{n-1}$ can be reconstructed by the induction hypothesis. Thus, to reconstruct $A$ from $A_1$ and $A_2$, it is enough to compare $A[1]$ and $A[n]$.

If the answer to either $\RMax(1, n, q)$ or $\RMin(1, n, q)$ includes position $1$ or $n$, then we are done. This also covers the case where $A[1] = A[n]$, since equal elements in $A$ always appear in consecutive positions.

Otherwise, let $x$ and $y$ be the rightmost positions of the smallest and largest elements in $A$, respectively, which can be determined using $\RMin(1, n, q)$ and $\RMax(1, n, q)$. Note that in this case $x, y \in \{2, \dots, n - 1\}$ and $x \neq y$. Thus, $A[1]$, $A[x]$, $A[y]$, and $A[n]$ must all have distinct values. Without loss of generality, assume $x < y$ (the other case is symmetric). 

Since $A[1]$, $A[x]$, $A[y]$, and $A[n]$ are all distinct, and any subsequence of $A$ with distinct values forms a Baxter permutation, we can apply the same argument as in Lemma 3 of~\cite{DBLP:conf/icalp/GawrychowskiN15}, which shows that either (i) there exists a position $i \in [x, y]$ such that $A[1] < A[i] < A[n]$ or $A[1] > A[i] > A[n]$, or (ii) $A[1] < A[n]$, thereby proving the theorem. 
\end{proof}
\end{theorem}

\section{Conclusion}
This paper proposes an encoding data structure that efficiently supports range ($q$-th) minimum/maximum queries and next/previous larger/smaller value queries. Our results match the current best upper bound of Tsur~\cite{DBLP:journals/ipl/Tsur19} up to lower-order additive terms while supporting the queries efficiently.

Note that the lower bound of Theorem~\ref{thm:lb} only considers the case that the same elements always appear consecutively, which still gives a gap between the upper and lower bound of the space. Improving the lower bound of the space for answering the queries would be an interesting open problem. 
In addition, to the best of our knowledge, there is currently no practical implementation of data structures that support queries based on both (colored) 2D-min and 2D-max heaps simultaneously. Implementing and optimizing the theoretical results for these problems is another promising direction for future work.



\bibliography{ref}



\end{document}